\documentclass[11pt]{article}
\usepackage{verbatim}
\usepackage{amsfonts}
\usepackage{graphics}
\usepackage{amsmath}
\usepackage{times}
\usepackage{appendix}
\usepackage{graphicx}
\usepackage{color}
\usepackage{enumerate}
\usepackage{fancyhdr,latexsym,amsmath,amsfonts,amssymb,amsbsy,amsthm,url}
\usepackage[margin=0.5in,footskip=0.25in]{geometry}
\usepackage{graphics,graphicx,epsfig}
\usepackage{breqn}
\usepackage{caption,subcaption,float,subfloat}
\newtheorem{theorem}{Theorem}[]

\usepackage[english]{babel}
\usepackage[dvipsnames]{xcolor}
\usepackage{tikz}
\usepackage{ulem}
\usepackage{soul}
\setstcolor{red}

\numberwithin{equation}{section}

\fancypagestyle{usual}{
  \lhead[\thepage]{}
  \chead[{\it \shortauthors}]{\sc \shorttitle}
  \rhead[]{\thepage}
  \lfoot[]{}
  \cfoot[]{}
  \rfoot[]{}
  
}

\makeatletter

\renewcommand{\section}{
  \@startsection
  {section}
  {1}
  {0pt}
  {1.1\baselineskip}
  {0.2\baselineskip}
  {\sc \centering}
}

\renewcommand{\subsection}{
  \@startsection
  {subsection}
  {1}
  {0pt}
  {1.1\baselineskip}
  {0.2\baselineskip}
  {\sc \centering}
}

\renewcommand{\subsubsection}{
  \@startsection
  {subsubsection}
  {1}
  {0pt}
  {1.1\baselineskip}
  {0.2\baselineskip}
  {\sc \centering}
}

\makeatother

\begin{document}

\begin{center}
{\large \sc Progression, Detection and Remission: Evolution of Chronic Myeloid Leukemia using a three-stage probabilistic model}\\
\vspace{0.4in}
{\large \sc Sonjoy Pan $^{\dagger}$}, 
{\large \sc Siddhartha P. Chakrabarty $^{\ast}$} 
and {\large \sc Soumyendu Raha $^{\P}$}\\
{$^{\dagger}$ Department of Mathematics, Swami Vivekananda University, Kolkata-700121, India}\\
{$^{\ast}$ Department of Mathematics, Indian Institute of Technology Guwahati, Guwahati-781039, India}\\
{$^{\P}$ Department of Computational and Data Sciences, Indian Institute of Science, Bangalore-560012, India}\\
\textit{$^{\dagger}$ sonjoy.pan@alumni.iitg.ac.in},
\textit{$^{\ast}$ pratim@iitg.ac.in},
and 
\textit{$^{\P}$ raha@iisc.ac.in}
\end{center}

\begin{abstract}

We present a three-stage probabilistic model for the progression of Chronic Myeloid Leukemia (CML), as manifested by the leukemic stem cells, progenitor cells and mature leukemic cells. This progression is captured through the process of cell division and cell mutation, with probabilities of occurrence being assigned to both of them. The key contributions of this study include, the determination of the expected number of the leukemic stem cells, progenitor cells, mature leukemic cells, as well as total number of these cells (in terms of probabilities, and contingent on the initial cell count), expected time to reach a threshold level of total and injurious leukemic cells, as well as the critical time when the disease changes its phases, the probability of extinction of CML, and the dynamics of CML evolution consequent to primary therapy. Finally, the results obtained are demonstrated with numerical illustrations.


\end{abstract}

\section{Introduction}
\label{Project_01_Section_Introduction}

The occurrence of leukemia is triggered by a disruption of the highly regulated and complex process of hematopoiesis, resulting from the generation of a handful of mutated blood cells \cite{Clapp16}. Leukemia can be acute or chronic and myeloid or lymphocytic, contingent on the maturity state and type of cells, respectively. Chronic Myeloid Leukemia (CML), a commonly occurring type of leukemia, is typically attributed to the translocation between chromosomes 9 and 22, resulting in a longer chromosome 9, and a shorter chromosome 22, than the typical chromosomes \cite{Pujo-Menjouet04}. CML, which is characterized by elevated levels of white blood cells, accounts for a fifth of the total leukemia incidences, of which 90 percent cases are associated with the Philadelphia (Ph) chromosome, resulting from the translocation of chromosomes 9 and 22. The typical progression of CML occurs in three stages \cite{Komariva09,Koch17}, namely, the chronic phase (a period of relatively low cell count, accompanied by a high level of cellular differentiation), the accelerated phase (during which higher cell levels is observed, along with decline in the extent of differentiation), and the blast phase (where an effectively uncontrolled cell growth, but low extent of differentiation, is observed). In absence of intervention, the disease can progress to the accelerated phase over a span of 7-10 years \cite{Koch17}, which then progresses to the blast phase.

Michor et al. \cite{Michor05}, in a basic model, proposed to capture the dynamics and progression of CML, and considered an ordinary differential equation (ODE) driven setup for normal and leukemic (non-resistant and resistant) cells, for a four-compartment model of stem cells, progenitor cells, differentiated cells and terminally differentiated cells. The model evaluated the success of imatinib during the molecularly targeted therapy in case of CML, and concluded that, while imatinib achieves success in case of inhibition of differentiated leukemic cells, it does not cause a decline in the levels of leukemic stem cells. Further, the authors also make an estimation of the likelihood of development of resistance to imatinib, as well as the time-point at which the resistance is detected. It is believed that the proliferation of cancerous stem cells, and progenitor cells brings forth, and abets the cellular growth in CML \cite{Komariva07}. In this paper \cite{Komariva07}, a stochastic model is considered, with primitive and proliferating, as well as quiescent CML cell populations, to explain the clinically observed data exhibiting, bi-phasic, one-phasic and reverse bi-phasic decline pattern, as also the emanation of quiescent stem cells and treatment resistance. An agent based model (ABM) to capture the ``age-structure'' of CML stem cells, is proposed in \cite{Roeder09} through the usage of normal hematopoiesis, leukemia genesis (for normal and leukemia cells) and leukemia treatment (using imatinib). A mathematical model considered in \cite{Radulescu14}, included the leukemic and normal cells, both stem-like and differentiated, leading to a four-dimensional coupled delay differential equation (DDE) system. An ODE driven model to analyze the dynamics of imatinib therapy is developed in \cite{Clapp15}, with the focus being on patients exhibiting a non-monotonic BCR-ABL ratio, during the course of the therapy, by considering mature cells (quiescent, cycling, progenitor and mature), along with the concentration of immune cells. A simplified version of the model in \cite{Clapp15} was considered in \cite{Besse18}, based on the argument that a three-compartment model (two for leukemic cells and one for immune cells) can capture the dynamics of the CML-immune response.

Population of hematopoietic stem cells and differentiated cells, both normal and cancerous, were considered in a compartment model in \cite{Ainseba10}, in the paradigm of an optimal control problem, with the objective of minimizing the treatment cost and load of the cancerous cells. In another control theoretic problem \cite{Radulescu16}, a two population model, involving short-term hematopoietic stem cells and differentiated cells, possessing the ability to self-renew, and lacking the same, respectively, is considered, in presence of imatinib. In a recent work on optimization of combination therapy of leukemia, of imatinib and interferon-$\alpha$, populations of CML cells and effector T cells cytotoxic to CML, are incorporated into the model \cite{Bunimovich-Mendrazitsky17}.

A stochastic approach to address the problem of acquired drug resistance to imatinib, and the consequent disease progression was considered in \cite{Leder11}, wherein a multi-type branching process was used to encapsulate the evolution of several clones of CML cells, resistant to imatinib. The model for hematopoiesis in CML, emphasized the stochastic aspects applicable in case of abnormalities in hematopoiesis, that is typical in cases of blood cancers. To this end, stochastic compartments and stochastic competition were taken into account in the hybrid discrete-continuous modeling framework. A causal Bayesian network model for explaining the progression of CML from the chronic phase to blast crisis is presented in \cite{Koch17}. In a recent article, the authors presented deterministic and stochastic models for a triple treatment scenario of Wnt/$\beta$-catenin signaling, combined with Tyrosine Kinase Inhibitors (TKI) and IFN-$\alpha$ \cite{Pan20}.

We also briefly review here, some of the literature on usage of branching process in epidemiological and cancer modeling. An epidemiological model for determination of the number of generations until extinction (for a branching process), in the paradigm of infectious disease modeling was discussed in \cite{Farrington99}. Accordingly, several processes were used to obtain the generation distributions, and reconciled in terms of application, to several infectious disease settings. In \cite{Durrett10}, the modeling was carried out for a population of cancerous cells, exhibiting exponential growth, progressing to treatment resistance, as a result of one mutation, followed by development of disease phenotype, after subsequent mutations. From a stochastic perspective, the distribution of the arrival time of the $k$-th mutation, as well as the growth progression of the count of type $k$ cells, was established. A random fitness increments approach, to model the progression of tumors was studied in \cite{Durrett10b}. A multi-type branching process was analyzed for accumulated mutations in cells, chosen from some distribution. An ovarian cancer model, with the objective of ascertaining an opportune window for screening was discussed in \cite{Danesh12}. The model is driven by a branching process, to describe the growth and progression of three cellular populations, namely, primary, peritoneal, and metastatic, with the model parameters being reconciled with the clinical data.

From a therapeutic perspective, TKIs, such as imatinib can have curative impact on CML, albeit the longer periods of treatment being required to achieve a desired result \cite{Lenaerts10}. A mathematical model focused on the combination treatment of imatinib (for molecular targeted therapy) and interferon alfa-2a (IFN-$\alpha$) (immunotherapy) presented in \cite{Berezansky12} incorporated time-varying delays in treatment terms, and demonstrated how the dual therapy led to improved outcomes in case of CML patients.

In this paper, we study on the dynamics of CML evolution, considering various stages of disease progression and different types of cells involved in leukemic cancer. We also determine the expected time to reach a possible state of CML and the chance of disease extinction. The rest of the paper is organized as follows. In Section \ref{Project_01_Section_Model}, the probabilistic modeling of CML progression is described, using the theory of multi-type branching processes. Some analytical results on cells count, estimated time, and associated probabilities for occurrence of various kinds of leukemic cells and subsequent conditions are established in Section \ref{Project_01_Section_Model_Analysis}. Section \ref{Project_01_Section_Results} is devoted to several numerical simulations. Finally, the conclusion is summarized in Section \ref{Project_01_Section_Conclusion}.

\section{Model}
\label{Project_01_Section_Model}

In CML patients, the leukemic cells develop through various cell differentiation and cell mutations. Initially, the leukemic stem cells are produced in the bone marrow, and then these cells get transformed to progenitor cells, which finally differentiate to mature leukemic cells, which are responsible for leukemic cancer. As observed earlier, the CML typically progresses in three phases, namely, chronic phase, accelerated phase and blast phase. The chronic phase is the initial stage of CML, where the myeloblasts capture up to $10\%$ of blood cells or bone marrow, and these increase to $10\%-19\%$ in the accelerated phase. The blast phase (acute phase or blast crisis) is the severe stage where the blood cells or bone marrow contain more than $20\%$ myeloblasts. Using the theory of multi-type branching processes, we analyze the three stages of CML progression based on three different types of leukemic cells, namely, stem cells, progenitor cells and mature leukemic cells. By calculating the percentage of total leukemic cells in blood, this study also helps in estimating the time when CML progresses from one phase to another (chronic, accelerated or blast).

We categorize the cells produced in the three stages of CML to be denoted as Type-$i$ cells, for $i=0,1,2$. Let us assume that a cell of these three types, generate new cells, by following the Poisson distribution, with parameter $\lambda$, which are then called the offsprings of the parent cell. The offspring mean denoted by $\lambda$ is the same for all types of cells. Further, assume that each Type-$i$ cell ($i=0,1$) generates both Type-$i$ cells (through cell division) and Type-($i+1$) cells (through cell mutation), but a Type-$2$ cell generates only Type-$2$ cells (through cell division). Let $p$ be the probability of the birth of the new cell through mutation, and hence $(1-p)$ is the probability of cell division for Type-$i$ cells, $i=0,1$. Further, the assumption is that each of the three types of cells either generates new cells (through cell division or cell mutation) with probability $a$ after time $\Delta t$, or dies within the interval $\Delta t$ with probability $(1-a)$. Thus $a$ is the survival probability for each cell. When a cell generates progenies, it bursts and hence after generating some new cells, it dies. We take $\Delta t$ to be a unit time in which this happens. Let $Z_{i}(t)$ be the number of Type-$i$ cells at time $t$. Suppose that the development of CML can be traced back to the initial source of $s$ number of Type-$0$ cells. We denote $X_{i}^{(k)}(t)$ and $Y_{i}^{(k)}(t)$ to be the number of cells produced by the cell division and cell mutation, respectively, from  the $k$-th individual of Type-$i$ cells, at time $t$. Accordingly, we obtain,
\begin{eqnarray*}
Z_{0}(0)&=&s,\\
Z_{0}(t+1)&=&\sum\limits_{k=1}^{Z_{0}(t)}X_{0}^{(k)}(t),\\
Z_{1}(t+1)&=&\sum\limits_{k=1}^{Z_{0}(t)}Y_{0}^{(k)}(t)+\sum\limits_{k=1}^{Z_{1}(t)}X_{1}^{(k)}(t),\\
Z_{2}(t+1)&=&\sum\limits_{k=1}^{Z_{1}(t)}Y_{1}^{(k)}(t)+\sum\limits_{k=1}^{Z_{2}(t)} X_{2}^{(k)}(t).
\label{EQ_01}
\end{eqnarray*}
We assume that each individual of Type-$0$ cells gives birth to a random number of Type-$1$ cells denoted by $X$. Hence $X$ is a random variable which follows the Poisson distribution, with parameter $\lambda$, and this is called offspring distribution, with $\lambda$ being the offspring mean. Suppose $X_{i}^{(k)}$ and $Y_{i}^{(k)}$ are independent identically distributed (i.i.d.) random variable with the same probability distribution as $X$. 
If $k$-th individual of Type-$i$ cells has $U_{i}^{(k)}$ offspring at time $t$, where $U_{i}^{(k)}$'s are i.i.d. random variable, then after $(t+1)$ time, Type-$i$ cells have total of
$\displaystyle{\sum\limits_{k} U_{i}^{(k)}}$ offspring, which is also a Poisson random variable with parameter ($\displaystyle{\lambda\times}$ the number of Type-$i$ cells).
The schematic representation of the model is presented in Figure
\ref{Schematic_Diagram}. 

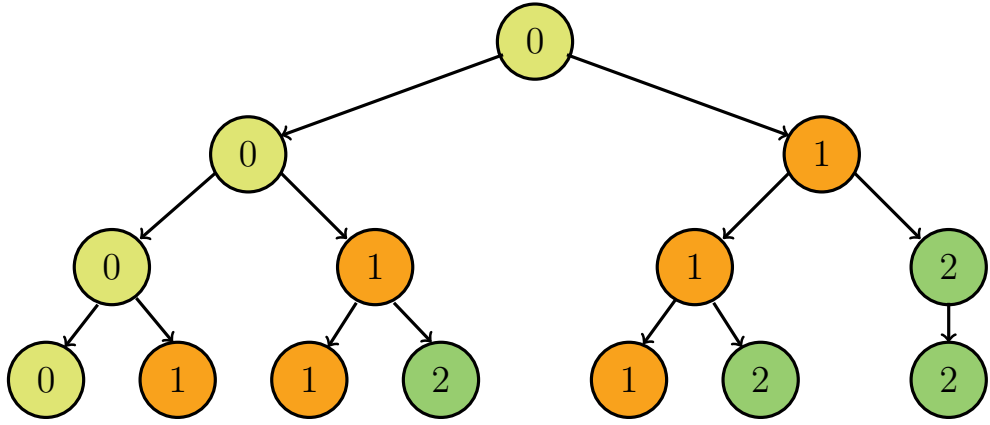
\begin{figure}[h]
\centering{
\begin{tikzpicture}[scale=0.25]
\draw[GreenYellow,fill] (0,0) circle (2cm);
\draw[very thick] (0,0) circle (2cm);
\draw[GreenYellow,fill] (-15.25,-6) circle (2cm);
\draw[very thick] (-15.25,-6) circle (2cm);
\draw[YellowOrange,fill] (15.25,-6) circle (2cm);
\draw[very thick] (15.25,-6) circle (2cm);
\draw[GreenYellow,fill] (-22.5,-12) circle (2cm);
\draw[very thick] (-22.5,-12) circle (2cm);
\draw[YellowOrange,fill] (-8.5,-12) circle (2cm);
\draw[very thick] (-8.5,-12) circle (2cm);
\draw[YellowOrange,fill] (8.5,-12) circle (2cm);
\draw[very thick] (8.5,-12) circle (2cm);
\draw[YellowGreen,fill] (22,-12) circle (2cm);
\draw[very thick] (22,-12) circle (2cm);
\draw[GreenYellow,fill] (-26,-18) circle (2cm);
\draw[very thick] (-26,-18) circle (2cm);
\draw[YellowOrange,fill] (-19,-18) circle (2cm);
\draw[very thick] (-19,-18) circle (2cm);
\draw[YellowOrange,fill] (-12,-18) circle (2cm);
\draw[very thick] (-12,-18) circle (2cm);
\draw[YellowGreen,fill] (-5,-18) circle (2cm);
\draw[very thick] (-5,-18) circle (2cm);
\draw[YellowOrange,fill] (5,-18) circle (2cm);
\draw[very thick] (5,-18) circle (2cm);
\draw[YellowGreen,fill] (12,-18) circle (2cm);
\draw[very thick] (12,-18) circle (2cm);
\draw[YellowGreen,fill] (22,-18) circle (2cm);
\draw[very thick] (22,-18) circle (2cm);
\draw[very thick,->] (-1.7,-0.7) -- (-13.5,-5);
\draw[very thick,->] (1.7,-0.7) -- (13.5,-5);
\draw[very thick,->] (-17,-7) -- (-21,-10.5);
\draw[very thick,->] (-13.5,-7) -- (-10,-10.5);
\draw[very thick,->] (13.5,-7) -- (10,-10.5);
\draw[very thick,->] (17,-7) -- (20.5,-10.5);
\draw[very thick,->] (-23.25,-14) -- (-25,-16.25);
\draw[very thick,->] (-21.2,-13.65) -- (-19.25,-16);
\draw[very thick,->] (-9.5,-13.9) -- (-11,-16.25);
\draw[very thick,->] (-7.5,-13.9) -- (-5.5,-16);
\draw[very thick,->] (7.5,-13.7) -- (5.75,-16.1);
\draw[very thick,->] (9.5,-13.9) -- (11,-16.25);
\draw[very thick,->] (22,-13.9) -- (22,-16);
\draw[ultra thick] (0,-1.25) node[anchor=south]{{\Large $0$}};
\draw[ultra thick] (-15.25,-7.25) node[anchor=south]{{\Large $0$}};
\draw[ultra thick] (15.25,-7.25) node[anchor=south]{{\Large $1$}};
\draw[ultra thick] (-22.5,-13.25) node[anchor=south]{{\Large $0$}};
\draw[ultra thick] (-8.5,-13.25) node[anchor=south]{{\Large $1$}};
\draw[ultra thick] (8.5,-13.25) node[anchor=south]{{\Large $1$}};
\draw[ultra thick] (22,-13.25) node[anchor=south]{{\Large $2$}};
\draw[ultra thick] (-26,-19.25) node[anchor=south]{{\Large $0$}};
\draw[ultra thick] (-19,-19.25) node[anchor=south]{{\Large $1$}};
\draw[ultra thick] (-12,-19.25) node[anchor=south]{{\Large $1$}};
\draw[ultra thick] (-5,-19.25) node[anchor=south]{{\Large $2$}};
\draw[ultra thick] (5,-19.25) node[anchor=south]{{\Large $1$}};
\draw[ultra thick] (12,-19.25) node[anchor=south]{{\Large $2$}};
\draw[ultra thick] ((22,-19.25) node[anchor=south]{{\Large $2$}};
\end{tikzpicture}
}
\caption{Schematic Representation of the Model.}
\label{Schematic_Diagram}
\end{figure}

\section{Model Analysis}
\label{Project_01_Section_Model_Analysis}

In this section,  we carry out the model analysis in detail. We begin with the estimation of the total number of offspring (total progeny) after development of leukemic cancer cells and calculate the time duration of the process. In the discussion, we consider $\lambda a>1$, as the number of CML cells initially increases over time.
\begin{theorem}
\label{Theorem_Total_Cells}
The expected number of total leukemic cells at time $t$ is given by,
\[E\left[Z(t)|Z(0)=s\right]=s\lambda^{t}a^{t},~t\ge 0.\]
\end{theorem}
\begin{proof}
We begin with $Z_{0}(0)=s$ cells. Now, recall that the probability of survival and the offspring mean for each cell, in a unit interval of time, are given by $a$ and $\lambda$, respectively. Hence, each of the $s$ cells produces $\lambda a$ offsprings, on an average, after a unit interval of time. Therefore, at time $t=1$, the expected number of CML cells (both Type-$0$ and Type-$1$ cells) becomes $s\lambda a$. Using a similar argument, at time $t=2$, the expected number of CML cells (Type-$0$, Type-$1$ and Type-$2$ cells) is given by $s\lambda^{2} a^{2}$. Proceeding in the same manner, we can conclude that the expected number of total leukemic cells (all three types of cells) at time $t$ will be $s\lambda^{t}a^{t}$.
\end{proof}	

\begin{theorem}
\label{Theorem_Type_0_Cells}	
The expected number of Type-$0$ cells at time $t$ is given by,
\[E[Z_{0}(t)|Z(0)=s]=s\lambda^{t}a^{t}(1-p)^{t},~t\ge 0.\]
\end{theorem}
\begin{proof}
We begin with $Z_0(0)=s$ cells. Further, $a$ is the probability of survival and $\lambda$ is the offspring mean for a cell, in a unit interval of time, with $(1-p)$ being the probability of cell division. Hence, each of Type-$0$ cells produces $\lambda a(1-p)$ off-springs on an average, after a unit interval of time. Hence, at time $t=1$, the expected number of Type-$0$ cells comes out to be $s\lambda a(1-p)$. Similarly, at time $t=2$, the expected number of Type-$0$ cells is given by $s\lambda^{2}a^{2}(1-p)^{2}$. Proceeding in the same way, we can conclude that the expected number of Type-$0$ cells at time $t$ will be $s\lambda^{t}a^{t}(1-p)^{t}$.	
\end{proof}	

\begin{theorem}
\label{Theorem_Type_1_Cells}	
The expected number of Type-$1$ cells at time $t$ is given by,
\[E[Z_{1}(t)|Z(0)=s]=ts\lambda^{t}a^{t}p(1-p)^{t-1},~t\ge 1.\]	
\end{theorem}
\begin{proof}
We use the principle of mathematical induction in order to prove the theorem. Initially, at time $t=0$, there are no Type-$1$ cell. At time $t=1$, the expected number of Type-$1$ cells is given by $s\lambda ap$, which comes only from Type-$0$ cells, through the process of cell mutation. Then Type-$1$ cells will be produced through both the processes, namely, the cell mutation from Type-$0$ cells, and the cell division from Type-$1$ cells. Therefore, the number of Type-$1$ cells at time $t=2$ depends on the number of Type-$0$ cells, as well as the number of Type-$1$ cells at $t=1$. At time $t=2$, the number of Type-$1$ cells produced from Type-$0$ cells is $s\lambda^{2}a^{2}p(1-p)$, which is the same as the number of Type-$1$ cells produced from Type-$1$ cells. Therefore, the expected number of Type-$1$ cells at time $t=2$ is equal to $2s\lambda^{2}a^{2}p(1-p)$.

Let us now assume that the theorem holds for $t=m$. Hence, 
\[E[Z_{1}(m)|Z(0)=s]=ms\lambda^{m}a^{m}p(1-p)^{m-1}.\] 
Now, the value of $\displaystyle{E[Z_{1}(m+1)|Z(0)=s]}$ is dependent on $\displaystyle{E[Z_{0}(m)|Z(0)=s]}$ and $\displaystyle{E[Z_{1}(m)|Z(0)=s]}$. Therefore, we have,
\[E[Z_{0}(m)|Z(0)=s]=s\lambda^{m}a^{m}(1-p)^{m}.\] 
The number of Type-$1$ cells produced from $\displaystyle{s\lambda^{m}a^{m}(1-p)^{m}}$ Type-$0$ cells, through the cell mutation is $\displaystyle{s\lambda^{m+1}a^{m+1}p(1-p)^{m}}$. Further, the number of Type-$1$ cells produced from $\displaystyle{ms\lambda^{m}a^{m}p(1-p)^{m-1}}$ Type-$1$ cells through the cell division is\\ 
$\displaystyle{ms\lambda^{m+1}a^{m+1}p(1-p)^{m}}$. Therefore, we finally obtain,  \[E[Z_{1}(m+1)|Z(0)=s]=(m+1)s\lambda^{m+1}a^{m+1}p(1-p)^{m}.\] Hence, by the principle of mathematical induction, the Theorem holds.
\end{proof}	

\begin{theorem}
\label{Theorem_Type_2_Cells}	
The expected number of Type-$2$ cells at time $t$ is given by,
\[E[Z_{2}(t)|Z(0)=s]=s\lambda^{t}a^{t}p^{2}
\left[1+2(1-p)+3(1-p)^{2}+\dots+(t-1)(1-p)^{t-2}\right],
~t\ge 2,\] 
or, 
\[E[Z_{2}(t)|Z(0)=s]= s\lambda^{t}a^{t}\left[1-(1-p)^t-tp(1-p)^{t-1}\right],
~t\ge 2.\]
\end{theorem}
\begin{proof}
We use the principle of mathematical induction to prove the theorem. At $t=0$ and $t=1$, there is no Type-$2$ cell. At $t=2$, the expected number of Type-$2$ cells is given by $\displaystyle{s\lambda^{2}a^{2}p^{2}}$, which comes only from $\displaystyle{s\lambda a(1-p)}$ Type-$1$ cells through the process of cell mutation. Then, the Type-$2$ cells will be produced through both the processes, namely, the cell mutation from Type-$1$ cells, and the cell division from Type-$2$ cells. Therefore, the number of Type-$2$ cells at time $t=3$ depends on the number of Type-$1$ cells as well as Type-$2$ cells at time $t=2$. At time $t=3$, the number of Type-$2$ cells that are produced from $2s\lambda^{2}a^{2}p(1-p)$ Type-$1$ cells is $2s\lambda^{3}a^{3}p^{2}(1-p)$ and those which are produced from $\displaystyle{s\lambda^{2} a^{2}p^{2}}$ Type-$2$ cells is  $\displaystyle{s\lambda^{3}a^{3}p^{2}}$. Therefore, the expected number of Type-$2$ cells at time $t=3$ is equal to $\displaystyle{s\lambda^{2}a^{2}p^{2}+s\lambda^{3} a^{3}p^{2},~\textit{i.e.,}~s\lambda^{3}a^{3}p^{2}[1+2(1-p)]}$.
Let us consider that the theorem holds for $t=m$. Hence, \[E[Z_{2}(m)|Z(0)=s]=s\lambda^{m}a^{m}p^{2}\left[1+2(1-p)+3(1-p)^{2}+\dots+(t-1)(1-p)^{m-2}\right].\] 
Now, the value of $\displaystyle{E[Z_{2}(m+1)|Z(0)=s]}$ is dependent on $\displaystyle{E[Z_{1}(m)|Z(0)=s]}$	and $\displaystyle{E[Z_{2}(m)|Z(0)=s]}$. \\We have $\displaystyle{E[Z_{1}(m)|Z(0)=s]=ms\lambda^{m}a^{m}p(1-p)^{m-1}}$. 
The number of Type-$2$ cells produced from $\displaystyle{ms\lambda^{m}a^{m}p(1-p)^{m-1}}$ Type-$1$ cells through the cell mutation is $\displaystyle{ms\lambda^{m+1}a^{m+1}p^{2}(1-p)^{m-1}}$. Further, the number of Type-$2$ cells produced from $\displaystyle{s\lambda^{m}a^{m}p^{2}\left[1+2(1-p)+3(1-p)^{2}+\dots+(t-1)(1-p)^{m-2}\right]}$ Type-$2$ cells through the cell division is\\ $\displaystyle{s\lambda^{m+1}a^{m+1}p^{2}\left[1+2(1-p)+3(1-p)^{2}+\dots+(t-1)(1-p)^{m-2}\right]}$. 
Therefore, 
\begin{eqnarray*}
E[Z_{1}(m+1)|Z(0)=s]&=&ms\lambda^{m+1}a^{m+1}p^{2}(1-p)^{m-1}\\
&+&s\lambda^{m+1}a^{m+1}p^{2}\left[1+2(1-p)+3(1-p)^{2}+
\dots+(t-1)(1-p)^{m-2}\right]\\
&=&s\lambda^{m+1}a^{m+1}p^{2}
\left[1+2(1-p)+3(1-p)^{2}+\dots+(m-1)(1-p)^{m-2}+m(1-p)^{m-1}\right]
\end{eqnarray*}
Hence, by the principle of mathematical induction, the Theorem holds. Furthermore, we obtained in Theorem \ref{Theorem_Total_Cells} that
$E[Z(t)|Z(0)=s]=E[Z_{0}(t)|Z(0)=s]+E[Z_{1}(t)|Z(0)=s]+E[Z_{2}(t)|Z(0)=s]=s\lambda^{t}a^{t}.$ 
Hence, \[E[Z_{2}(t)|Z(0)=s]=s\lambda^{t}a^{t}-E[Z_{0}(t)|Z(0)=s]-E[Z_{1}(t)|Z(0)=s]=s\lambda^{t}a^{t}\left[1-(1-p)^{t}-tp(1-p)^{t-1}\right],~t\ge 2.\]
\end{proof}

The results presented in Theorems \ref{Theorem_Total_Cells} to \ref{Theorem_Type_2_Cells} are presented in a summarized form in Table \ref{Table_01}.

\begin{table}[h]
\begin{center}
\begin{tabular}{ll}
\hline
$i$	 & Expected number of Type-$i$ cells \\ 
\hline
$i=0$ & $E[Z_{0}(t)|Z(0)=s]=s\lambda^{t}a^{t}(1-p)^{t},~t\ge 0$. \\
\hline
$i=1$ & $E[Z_{1}(t)|Z(0)=s]=ts\lambda^{t}a^{t}p(1-p)^{t-1},~t\ge 1$.\\
\hline
$i=2$ & $E[Z_{2}(t)|Z(0)=s]= s\lambda^{t}a^{t}\left[1-(1-p)^t-tp(1-p)^{t-1}\right],
~t\ge 2$.\\
\hline
\end{tabular}
\caption{\label{Table_01} Summary of Theorems \ref{Theorem_Total_Cells} to \ref{Theorem_Type_2_Cells} .}
\end{center}
\end{table}

\begin{theorem}
\label{Theorem_Expected_Time_T}
The time to reach the leukemia of $M$ cells on average is given by,
\[T=\frac{\ln\left(M/s\right)}{\ln(\lambda a)},~\lambda a>1.\]	
\end{theorem}
\begin{proof}
If $T$ be the time required to reach $M$ number of leukemic cells on average, then we can write, from Theorem \ref{Theorem_Total_Cells},
\[s\lambda^{T}a^{T}=M,\]
which gives,
\[T=\frac{\ln\left(M/s\right)}{\ln(\lambda a)},
~\lambda a>1.\]	
\end{proof}	

\begin{theorem}
\label{Theorem_Expected_Time_T12}
If $T_{12}$ be the time required to reach injurious leukemic cells of average size $M_{12}$, the following relation holds,
\[s\lambda^{T_{12}}a^{T_{12}}\left[1-(1-p)^{T_{12}}\right]=M_{12}.\]
\end{theorem}
\begin{proof}
Let $T_{12}$ be the time required to reach $M_{12}$ number of injurious leukemic cells. Then, from Theorem \ref{Theorem_Type_1_Cells} and Theorem \ref{Theorem_Type_2_Cells}, we arrive at the following,
\begin{eqnarray*}
M_{12}&=&T_{12}s\lambda^{T_{12}}a^{T_{12}}p(1-p)^{T_{12}-1}+s\lambda^{T_{12}}a^{T_{12}}\left[1-(1-p)^{T_{12}}-T_{12}p(1-p)^{T_{12}-1}\right] \\
&=&s\lambda^{T_{12}}a^{T_{12}}\left[1-(1-p)^{T_{12}}\right].
\end{eqnarray*}  
\end{proof} 

Finally, we determine the probability of CML extinction, under the model considerations. The event $\{Z(t)=0,~\text{for some}~t> 0 ~|~Z(0)=s\}$ is called the disease extinction. Then we have the following Theorem.

\begin{theorem}
If $P(t)$ be the probability of extinction of CML at time $t$, then,
\[P(t)\ge(1-a)^{s\lambda^{t-1} a^{t-1}},~t\ge1.\]
\end{theorem}
\begin{proof}
According to Theorem \ref{Theorem_Total_Cells}, the expected number of total leukemic cells at time $t$ is $s\lambda^{t}a^{t}$. The disease CML becomes extinct at time $t$ if all leukemic cells produced till time $t-1$ die before $t-1$ or during $t-1$ to $t$. Hence, the probability of extinction of CML at time $t$ is given by,
\[P(t)\ge(1-a)^{s\lambda^{t-1} a^{t-1}},~t\ge1.\]
\end{proof} 

We now discuss about the consequences of the primary therapy on CML patients. The primary therapy for CML prevents the expansion of Type-0 cells, but it does not have a direct effect on Type-1 and Type-2 cells, since these cells are resistant to primary therapy. The effect on Type-0 cells, due to primary therapy results in changes in the development of Type-1 and Type-2 cells. Accordingly, we let the effect of the primary therapy be manifested through change in the probability of cell survival and cell mutation for Type-0 cells to $\bar{a}$ and $\bar{p}$, respectively, whereas those for Type-1 and Type-2 cells remains the same as prior to therapy. Then the results can be written in the following Theorems.
\begin{theorem}
\label{Theorem_Total_Cells_New}
If the primary CML therapy is initiated, the expected number of total leukemic cells at time $t$ is given by,	
\[E[\bar{Z}(t)|Z(0)=s]=s\lambda^{t}\bar{a}\left[\bar{p}a^{t-1}+\bar{p}a^{t-2}\{\bar{a}(1-\bar{p})\}+\bar{p}a^{t-3}\{\bar{a}(1-\bar{p})\}^{2}+\cdots+\bar{p}a\{\bar{a}(1-\bar{p})\}^{t-2}+\{\bar{a}(1-\bar{p})\}^{t-1}\right].\]
\end{theorem}
\begin{proof}
In order to prove this Theorem, we consider the following three facts:	
\begin{itemize}
\item Type-0 cells produce the same cells at a rate $\lambda$, with probability $1-\bar{p}$ (cell division), where the survival probability of Type-0 cells is $\bar{a}$.
\item Type-0 cells produce Type-1 cells at a rate $\lambda$, with probability $\bar{p}$ (cell mutation), where the survival probability of Type-0 cells is $\bar{a}$. 
\item Both Type-1 and Type-2 cells produce new leukemic cells of any type at a rate $\lambda$, with probability $1$ (cell division and mutation together), where the survival probability of Type-1 or Type-2 cells is $a$.
\end{itemize}
Therefore, the number of total leukemic cells at $t=n$ is dependent on the following three evaluations:
\begin{itemize}
\item Number of Type-0 cells produced from Type-0 cells, present at $t=n-1$, through cell division.
\item Number of Type-1 cells produced from Type-0 cells, present at $t=n-1$, through cell mutation. 
\item Total number of Type-1 and Type-2 cells produced from both of Type-1 and Type-2 cells, present at $t=n-1$, through cell division and mutation together.
\end{itemize}
\begin{enumerate}[(A)]
\item Calculation of total leukemic cells at $t=1$, 
\begin{itemize}
\item The number of Type-0 cells produced from Type-0 cells, present at $t=0$, through cell division is $s\lambda \bar{a}(1-\bar{p})$.
\item The number of Type-1 cells produced from Type-0 cells, present at $t=0$, through cell mutation is	$s\lambda\bar{a}\bar{p}$. 
\item There is no such type of leukemic cells which is produced from Type-1 or Type-2 cells, present at $t=0$. 
\end{itemize}
Therefore, at $t=1$, the number of total leukemic cells becomes $s\lambda \bar{a}(1-\bar{p})+s\lambda\bar{a}\bar{p}=s\lambda\bar{a}$.
\item Calculation of total leukemic cells at $t=2$, 
\begin{itemize}
\item The number of Type-0 cells produced from Type-0 cells, present at $t=1$, through cell division is $s\lambda^{2} \bar{a}^{2}(1-\bar{p})^{2}$.
\item The number of Type-1 cells produced from Type-0 cells, present at $t=1$, through cell mutation is	$s\lambda^{2}\bar{a}^{2}\bar{p}(1-\bar{p})$. 
\item Total number of Type-1 and Type-2 cells produced from both of Type-1 and Type-2 cells, present at $t=1$, through cell division and mutation together is $s\lambda^{2}a\bar{a}\bar{p}$. 
\end{itemize}
Therefore, at $t=2$, the number of total leukemic cells is,
\[s\lambda^{2}\bar{a}^{2}(1-\bar{p})^{2}+s\lambda^{2}\bar{a}^{2}\bar{p}(1-\bar{p})+s\lambda^{2}a\bar{a}\bar{p}
=s\lambda^{2}\bar{a}\left[\bar{p}a+\{\bar{a}(1-\bar{p})\}\right].\]
\item Calculation of total leukemic cells at $t=3$, 
\begin{itemize}
\item The number of Type-0 cells produced from Type-0 cells, present at $t=2$, through cell division is $s\lambda^{3}\bar{a}^{3}(1-\bar{p})^{3}$.
\item The number of Type-1 cells produced from Type-0 cells, present at $t=2$, through cell mutation is $s\lambda^{3}\bar{a}^{3}\bar{p}(1-\bar{p})^{2}$. 
\item Total number of Type-1 and Type-2 cells produced from both of Type-1 and Type-2 cells, present at $t=2$, through cell division and mutation together is $s\lambda^{3}a\bar{a}^{2}\bar{p}(1-\bar{p})+s\lambda^{3}a^{2}\bar{a}\bar{p}$. 
\end{itemize}
Therefore, at $t=3$, the number of total leukemic cells is, 
\begin{eqnarray*}
&&s\lambda^{3}\bar{a}^{3}(1-\bar{p})^{3}+s\lambda^{3}\bar{a}^{3}\bar{p}(1-\bar{p})^{2}+s\lambda^{3}a\bar{a}^{2}\bar{p}(1-\bar{p})+s\lambda^{3}a^{2}\bar{a}\bar{p}\\
&=&s\lambda^{3}\bar{a}\Big[\bar{p}a^{2}+\bar{p}a\{\bar{a}(1-\bar{p})\}+\{\bar{a}^{2}(1-\bar{p})^{2}\}\Big].
\end{eqnarray*}
\end{enumerate}
Proceeding in the same manner, one can obtain that the expected number of total leukemic cells at time $t$ is equal to, \[s\lambda^{t}\bar{a}\left[\bar{p}a^{t-1}+\bar{p}a^{t-2}\{\bar{a}(1-\bar{p})\}+\bar{p}a^{t-3}\{\bar{a}(1-\bar{p})\}^{2}+\cdots+\bar{p}a\{\bar{a}(1-\bar{p})\}^{t-2}+\{\bar{a}(1-\bar{p})\}^{t-1}\right].\]
\end{proof}	
\begin{theorem}
\label{Theorem_Type_0_Cells_New}	
If the primary CML therapy is initiated, the expected number of Type-$0$ cells at time $t$ is given by,
\[E[\bar{Z}_{0}(t)|Z(0)=s]=s\lambda^{t}{\bar{a}}^{t}(1-{\bar{p}})^{t}.\]
\end{theorem}
\begin{proof}
The proof is similar to Theorem \ref{Theorem_Type_0_Cells}.
\end{proof}	

\begin{theorem}
\label{Theorem_Type_1_Cells_New}	
If the primary CML therapy is initiated, the expected number of Type-$1$ cells at time $t$ is given by
\begin{eqnarray*}
E[\bar{Z}_{1}(t)|Z(0)=s]&=&s\lambda^{t}\bar{a}\bar{p}\left[\{a(1-p)\}^{t-1}+\{a(1-p)\}^{t-2}\{\bar{a}(1-\bar{p})\}+\{a(1-p)\}^{t-3}\{\bar{a}(1-\bar{p})\}^{2}+\cdots\right.\\
&&~+~\left.\{a(1-p)\}\{\bar{a}(1-\bar{p})\}^{t-2}+\{\bar{a}(1-\bar{p})\}^{t-1}\right].
\end{eqnarray*}
\end{theorem}
\begin{proof}
In order to prove this theorem, we consider the following two facts:	
\begin{itemize}
\item Type-0 cells produce Type-1 cells at a rate $\lambda$ with probability $\bar{p}$ (cell mutation), where the survival probability of Type-0 cells is $\bar{a}$. 
\item Type-1 cells produce the same cells at a rate $\lambda$ with probability $1-p$ (cell division), where the survival probability of Type-0 cells is $a$.
\end{itemize}
Therefore, the number of Type-1 cells at $t=n$ is dependent on the following two evaluations:
\begin{itemize}
\item Number of Type-1 cells produced from Type-0 cells, present at $t=n-1$, through cell mutation. 
\item Number of Type-1 cells produced from Type-1 cells, present at $t=n-1$, through cell division.
\end{itemize}
\begin{enumerate}[(A)]
\item Calculation of Type-1 cells at $t=1$, 
\begin{itemize}
\item The number of Type-1 cells produced from Type-0 cells, present at $t=0$, through cell mutation is $s\lambda\bar{a}\bar{p}$. 
\item There is no such Type-1 cells which is produced from Type-1 cells present at $t=0$. 
\end{itemize}
Therefore, at $t=1$, the number of Type-1 cells becomes $s\lambda\bar{a}\bar{p}$.
\item Calculation of Type-1 cells at $t=2$, 
\begin{itemize}
\item The number of Type-1 cells, produced from Type-0 cells present at $t=1$, through cell mutation is $s\lambda^{2}\bar{a}^{2}\bar{p}(1-\bar{p})$. 
\item The number of Type-1 cells produced from Type-1 cells, present at $t=1$, through cell division is $s\lambda^{2}a\bar{a}\bar{p}(1-p)$.
\end{itemize}
Therefore, at $t=2$, the number of Type-1 cells becomes $s\lambda^{2}\bar{a}^{2}\bar{p}(1-\bar{p})+s\lambda^{2}a\bar{a}\bar{p}(1-p)=s\lambda^{2}\bar{a}\bar{p}\Big[a(1-p)+\bar{a}(1-\bar{p})\Big]$.
\item Calculation of Type-1 cells at $t=3$, 
\begin{itemize}
\item The number of Type-1 cells produced from Type-0 cells, present at $t=2$, through cell mutation is	$s\lambda^{3}\bar{a}^{3}\bar{p}(1-\bar{p})^{2}$. 
\item The number of Type-1 cells produced from Type-1 cells, present at $t=2$, through cell division is $s\lambda^{3}a\bar{a}^{2}\bar{p}(1-p)(1-\bar{p})+s\lambda^{3}a^{2}\bar{a}\bar{p}(1-p)^{2}$.
\end{itemize}
Therefore, at $t=3$, the number of Type-1 cells is 
\begin{eqnarray*}
&&s\lambda^{3}\bar{a}^{3}\bar{p}(1-\bar{p})^{2}+s\lambda^{3}a\bar{a}^{2}\bar{p}(1-p)(1-\bar{p})+s\lambda^{3}a^{2}\bar{a}\bar{p}(1-p)^{2}\\
&=&s\lambda^{3}\bar{a}\bar{p}\Big[\{a(1-p)\}^{2}+\{a(1-p)\}\{\bar{a}(1-\bar{p})\}+\{\bar{a}(1-\bar{p})\}^{2}\Big].
\end{eqnarray*}
\end{enumerate}	
Proceeding in the same manner, one can obtain that the expected number of Type-1 cells at time $t$ equals, \[s\lambda^{t}\bar{a}\bar{p}\Big[\{a(1-p)\}^{t-1}+\{a(1-p)\}^{t-2}\{\bar{a}(1-\bar{p})\}+\{a(1-p)\}^{t-3}\{\bar{a}(1-\bar{p})\}^{2}+\cdots+\{a(1-p)\}\{\bar{a}(1-\bar{p})\}^{t-2}+\{\bar{a}(1-\bar{p})\}^{t-1}\Big].\]
\end{proof}	
\begin{theorem}
\label{Theorem_Type_2_Cells_new}	
If the primary CML therapy is initiated, the expected number of Type-$2$ cells at time $t$ can be calculated using the following relation,
\[E[\bar{Z}_{2}(t)|Z(0)=s]=E[\bar{Z}(t)|Z(0)=s]-E[\bar{Z}_{0}(t)|Z(0)=s]-E[\bar{Z}_{1}(t)|Z(0)=s].\]
\end{theorem}
\begin{proof}
The proof follows from Theorem \ref{Theorem_Total_Cells_New}, Theorem \ref{Theorem_Type_0_Cells_New} and Theorem \ref{Theorem_Type_1_Cells_New}.
\end{proof}	

\section{Results}
\label{Project_01_Section_Results}

In this penultimate Section, we present various illustrative numerical results, pertaining to the derivations made in Section \ref{Project_01_Section_Model_Analysis}. The levels of the three types of CML cells, as well as the total number of cells, as they progress over time are numerically illustrated. We also demonstrate how the number of CML cells can be impacted by the values of the parameters, such as the initial levels of the leukemic stem cells, the probability of survival of the CML cells, cell division or mutation rates, and the offspring mean. The simulation is carried out in order to observe the probability of extinction of the disease, with respect to the variation in the values of the parameters. Finally, we determine the optimal time, at which the progression of the disease changes from one state to another. For the purpose of some illustrative results, we choose the values of the parameters to be, $s=1$, $p=0.3$, $a=0.5$, $\lambda=3$, with the time window under consideration being $T=50$.

We begin with the results for the progression of the cells over time. Accordingly, Figure \ref{Fig_01} shows the changes in the levels of the three types of CML cells over time. We observe from the figure, that from the initial time of $t=0$, until the time of $t=40$, the leukemic stem cells and the progenitor cells, gradually increase with time, but there is no noticeable increment in the levels of the mature leukemic cells. However, beyond this point, there is a rapid and accelerated growth of the mature leukemic cells, thereby resulting in the host being predominantly invaded by these cells.

\begin{figure}[H]
\centering
\includegraphics[width=0.5\linewidth]{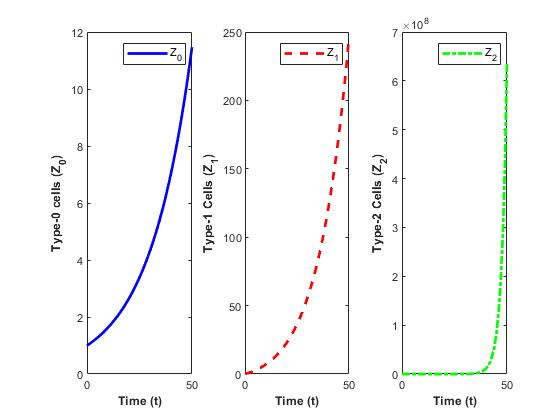}
\caption{Progression of the leukemic stem cells ($Z_{0}$), the progenitor cells ($Z_{1}$) and the mature leukemic cells ($Z_{2}$)}
\label{Fig_01}
\end{figure}

In Figure \ref{Fig_02}, we plot the percentage of injurious cells, defined as $\displaystyle{\frac{Z_{1}+Z_{2}}{Z_{0}+Z_{1}+Z_{2}}}\times 100$, as a function of time. It captures the progression of percentage of the progenitor cells and the mature leukemic cells, as a percentage of the total number of the leukemic cells, over time. As can be seen from Figure \ref{Fig_02}, the injurious cells, for all practical purpose have taken over the host after about time $t=12$, since the detection.

\begin{figure}[H]
\centering
\includegraphics[width=0.5\linewidth]{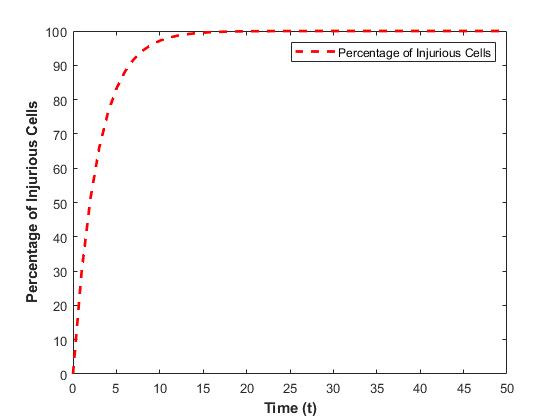}
\caption{Percentage of injurious cells as a function of $t$}
\label{Fig_02}
\end{figure}

Now, we present the three dimensional results for the progression of the three types of cells, against the ranges of $a\in [0,0.5]$ and $\lambda\in [0,3]$ in Figure \ref{Fig_03_04_05}. It can be observed that until one reaches fairly close to the maximum values in the ranges of $a$ and $\lambda$, the progression of the disease is reasonably under control, beyond which the cell levels dramatically reach fatal levels in the host.

\begin{figure}[H]
\begin{subfigure}{.3\textwidth}
\centering
\includegraphics[width=1\linewidth]{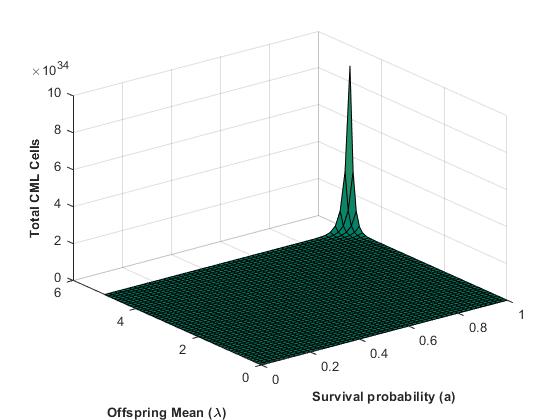}
\caption{$Z_{0}$ as a function of $a$ and $\lambda$}
\label{Fig_03}
\end{subfigure}
\begin{subfigure}{.3\textwidth}
\centering
\includegraphics[width=1\linewidth]{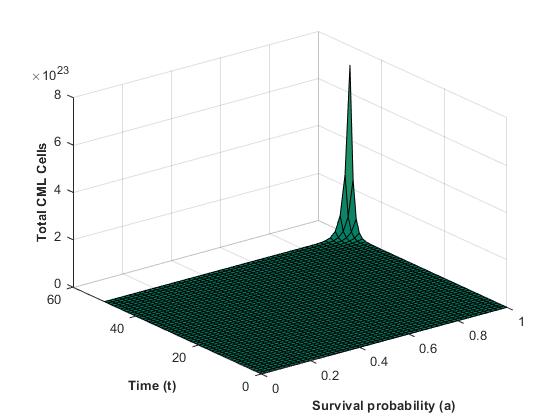}
\caption{$Z_{1}$ as a function of $a$ and $\lambda$}
\label{Fig_04}
\end{subfigure}
\begin{subfigure}{.3\textwidth}
\centering
\includegraphics[width=1\linewidth]{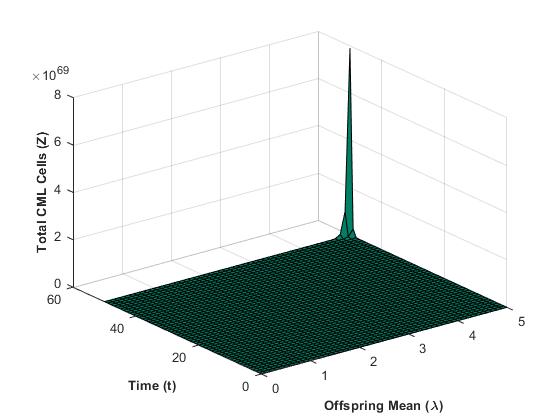}
\caption{$Z_{2}$ as a function of $a$ and $\lambda$}
\label{Fig_05}
\end{subfigure}
\caption{Progression of the leukemic stem cells ($Z_{0}$), the progenitor cells ($Z_{1}$) and the mature leukemic cells ($Z_{2}$) as a function of $a$ and $\lambda$}
\label{Fig_03_04_05}
\end{figure}

We next look at the probability of extinction for various parameter values in Figure \ref{Fig_06}. For the first figure, we choose the ranges of values of $\lambda$ to be $[0,3]$. The extinction is almost certain upto nearly the value of $\lambda=2$. However, beyond this point, there is a very small window of $\lambda$ (among the values of $\lambda a>1$) for which there is a positive probability of extinction, and then the probability is virtually nonexistent (demonstrating extremely poor prognosis).

\begin{figure}[H]
\centering
\includegraphics[width=0.5\linewidth]{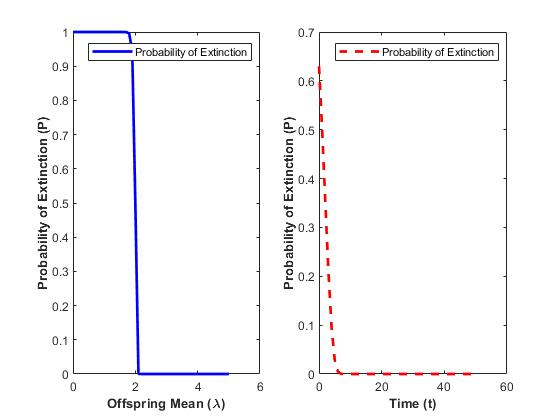}
\caption{Probability of extinction for various values of $\lambda$ and $t$}
\label{Fig_06}
\end{figure}

Finally, we look at the illustration for the results in Theorems \ref{Theorem_Expected_Time_T} and \ref{Theorem_Expected_Time_T12}. For the implementation of the result in Theorem \ref{Theorem_Expected_Time_T}, we consider the values of $0\le M\le 10^{8}$, and plot the corresponding expected times $T$ to reach a level of $M$ leukemic cells can be seen in Figure \ref{Fig_07}. This result helps in determining the expected time at which the CML disease changes from being the less injurious chronic phase, to the more serious and fatal condition based on the percentage of leukemic cells in the patient's blood. For the result in Theorem \ref{Theorem_Expected_Time_T12}, we consider the ranges $0\le T_{12} \le 45$ and $10^{2}\le M_{12} \le 10^{8}$, and present the implicit plot for the same in Figure \ref{Fig_08}. 

\begin{figure}[H]
\centering
\includegraphics[width=0.5\linewidth]{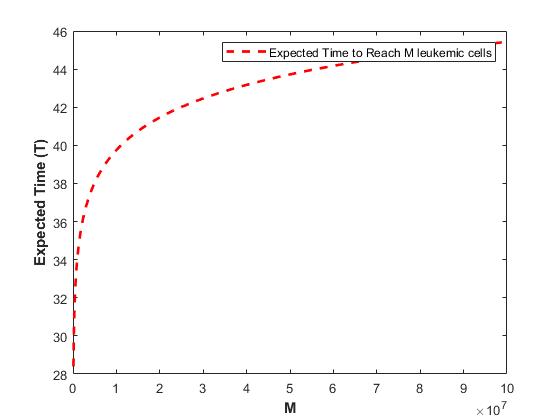}
\caption{Expected times $T$ to reach a level of $M$ leukemic cells}
\label{Fig_07}
\end{figure}

\begin{figure}[H]
\centering
\includegraphics[width=0.5\linewidth]{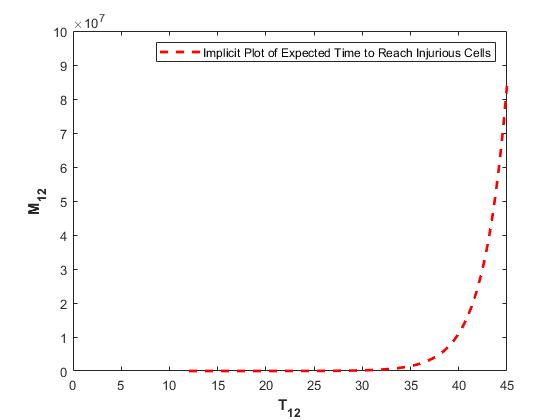}
\caption{Implicit plot of expected time $T_{12}$ to reach $M_{12}$ injurious leukemic cells }
\label{Fig_08}
\end{figure}

\section{Conclusion}
\label{Project_01_Section_Conclusion}

In this study, the evolution of CML is discussed  through various stages of disease progression with three types of leukemic cells. We have analytically calculated the expected number of total leukemic cells, as well as three types of leukemic cells, before and after initiation of the therapy. The expected times to reach a certain number of total leukemic cells and injurious leukemic cells are estimated. The probability of extinction of the CML is also analytically determined. The progression of different leukemic cells are graphically plotted which exhibited a sudden and large increment in Type-$2$ cells, as compared to other two types of cells. The numerical findings also suggests that after a while, the leukemic cells are dominated by the injurious cells. The changes in the number of leukemic cells, contingent on the values of survival probability and offspring mean are shown graphically. The chance of disease extinction is numerically shown for different values of offspring mean, which depicts that the chance of CML extinction is almost impossible without therapy, when the offspring mean exceeds $2$. Thus, the findings obtained can help in determining the percentage of leukemic cells in a patient's blood at a specific time and to diagnose the stage which the patient is going through. Also, it can help in estimating the time needed to go to the next stage of CML, and to predict the condition of the patient in advance, so that the medical practitioner can advice on the appropriate therapeutic intervention.

\section*{Acknowledgment}
\label{Project_01_Section_Acknowledgment}

SPC was supported by Grant. No. MTR/2019/000225 from the Science and Engineering Research Board, India. SR's work was supported partially by Grant No. MTR/2020/000186 from the Science and Engineering Research Board, India.

\end{document}